\documentclass[a4paper,10pt]{amsart}

\usepackage{graphicx}
\usepackage{epsfig}
\usepackage{subfig}
\usepackage{amsmath}
\usepackage{amssymb}
\usepackage{amsthm}
\usepackage[figuresright]{rotating}
\usepackage{bbm}

\newtheorem{theorem}{Theorem}[section]
\newtheorem{lemma}[theorem]{Lemma}

\theoremstyle{remark}
\newtheorem{remark}[theorem]{Remark}

\theoremstyle{definition}
\newtheorem{definition}[theorem]{Definition}

\newcommand{\dd}{\,\mathrm{d}}

\title{Higher Order Minimum Entropy Approximations in Radiative Transfer}
\author{Philipp Monreal}
\author{Martin Frank}
\address{Dept.\ of Mathematics\\Technical University of Kaiserslautern\\Erwin-Schr\"odinger-Str. \\67663 Kaiserslautern\\Germany}
\email{monreal@mathematik.uni-kl.de}
\email{frank@mathematik.uni-kl.de}
\begin{document}

\maketitle

\begin{abstract}
In this paper we approximate the radiative transfer equations by the method of moments, constructing mesoscopic approximations of arbitrary order of the otherwise microscopic system. To define the necessary closure a minimum entropy approach is utilized. While in radiative transfer, the minimum entropy closure for moment systems up to the first-order moment is well known, higher-order minimum entropy closures have not been implemented. This is probably due to the fact that the closure cannot be expressed in analytical form. 
Our focus thus lies in developing some general results about the minimum entropy system and in deriving a numerical closure. 
By extending to higher order, among increasing the precision, we are able to overcome difficulties that arise for the first order minimum entropy method. Numerical experiments in a 1-dimensional domain irradiated by two beams or with internal source show the accuracy of this approach.
\end{abstract}

\section{Introduction}
Radiation therapy has been in use for treatment of cancer for over a hundred years and has become one of its most important means. To predict the radiation dose, simulations solving an approximation of the radiative transfer equations are used. The radiative transfer equations are a system of integro-differential equations constituting a balance law for a radiation field in a medium. Currently, most applications in clinical use utilize so-called pencil-beam models \cite{ahnasp99}. However, they have fundamental problems handling inhomogeneities of the tissue \cite{krisau05,cygbatscrmahant87}, e.g.\ air cavities in lungs or the head, which are extremely vulnerable to irradiation. The most accurate methods of calculating the amount of absorbed radiative energy use Monte-Carlo models \cite{and91}. However, the necessary computation time often exceeds the capabilites in clinical environments.

A third approach is given by the method of moments. Originally developed by H. Grad in the context of rarefied gases \cite{grad49}, through integration over the directional variable one averages over all directions, introducing the moments of the particle distribution. A transition in scale is achieved: the mesoscopic system becomes macroscopic, which is a reasonable simplification, since we are interested in the radiation dose and not the exact distribution of all particles. By expanding into moments of arbitrary order a hierarchy of moment systems can be defined \cite{dubrofeu99}. A closure for the resulting system has to be defined in order to make it solvable. We choose minimum entropy closures which enforce that the distribution maximizes the physical entropy. The mathematical entropy to be minimized, hence the title, is just the physical one with a minus sign. This approach is inspired by the fact that physical systems always tend to increase the entropy. It has become the main concept of rational extended thermodynamics \cite{mulrug93}. Jaynes has shown that the entropy-minimizing distribution is indeed the most probable one \cite{jay57}. Subsequently, much work in this field has been done by Levermore, who proposed a closure ensuring the hyperbolicity of the system and its ability to dissipate the entropy locally \cite{levermore83,levermore96}. The moment system of order n, closed using the minimum entropy principle, is termed $\mathcal{M}_n$.

Unfortunately, higher order minimum entropy closures are not explicitly expressible, but have to be calculated numerically, which might be the main reason why they have not attracted much attention in practice so far. To describe the closure the so-called Eddington factor is used. By studying it, we will be able to gain insight into the system's behaviour, e.g.\ its ability to handle radiative non-equilibria. The one major drawback of the first order minimum entropy method $\mathcal{M}_1$ is that it cannot handle situations when a net particle flux of zero occurs. As will be shown, the $\mathcal{M}_2$ model is able to fix this problem. 

The remainder of this paper is organized as follows: in section \ref{chap:model_description} we introduce the underlying transport equations for electrons and photons separately and describe the method of moments, after explaining the physical background. In section \ref{chap:me_closure} we study the minimum entropy closure in general and the Eddington factor in detail to construct the closure of our system. The Eddington factor has to be calculated numerically, and we present on e method in section \ref{chap:computation_eddi}. Section \ref{chap:properties_Mn} deals with properties of different moment systems resulting from the method of moments and a minimum entropy closure. Finally, in section \ref{chap:num_results} numerical results of simulations of the transport equations are shown and discussed. 

\section{Model Description}
\label{chap:model_description}
\subsection{Physical background}
In radiation therapy, beams of protons, electrons or photons are used. The common point of interest is radiation dose, i.e.\ the amount of energy deposited in the tissue. Referring to interactions particles undergo within the body, we will distinguish between 3 different types: elastic scattering, inelastic scattering and absorption. Absorption is the process when a particle loses its energy completely. The former type of scattering, elastic scattering, describes processes in which a particle's energy remains unchanged, only the direction of flight changes. The latter one, inelastic scattering, characterizes the case where the particle loses a considerable amount of energy, oftentimes increasing the number of particles that contribute to their flux. The probability that a certain particle, interacting with an atom, is scattered by a certain angle and gains/loses a certain amount of energy is given by the differential scattering cross section, also called scattering kernel. As it is in our situation by far more probable that the deflection angle is very small, we call the scattering kernel forward-peaked. It is important to check whether the model is able to handle this anisotropy. In an opaque medium, the radiation is assumed to be near isotropic and macroscopic diffusion models can be used \cite{mihalas84,pom92}. In a transparent medium however, the radiation becomes strongly anisotropic where microscopic methods like Monte-Carlo simulations are used. As the radiation transfer in our present case belongs to the transitional regime between these two extremes, we try to find a model which is macroscopic but nevertheless capable of dealing with those non-equilibria. The energy of our particles is the relativistic kinetic energy in the case of protons or electrons and $E = h \cdot \nu$ in case of photons. A fundamental assumption in this work is that we only consider media stationary in space, since otherwise the transfer equation is coupled with an equation describing the fluid motion, as particles now would have a prefered traveling direction. This is reasonable, since it can be assumed that the patient remains stationary while the treatment is taking place.

\subsection{Transport equations}
In this section we state the transport equations used to describe the transport of electrons or that of photons moving through a medium. Let us start by introducing some fundamental notations: let $x$ be the spatial, $t$ be the time and $\Omega$ be the directional variable (in a spherical coordinate system): $\Omega = (\mu, \sqrt{1-\mu^{2}}\cos\varphi, \sqrt{1-\mu^2}\sin\varphi)^T$. $\epsilon$ describes the energy. It is also possible to consider the problem in slab (or planar) geometry: assume the medium has infinite length in two dimensions but only finite length in the third one. Imagine parallel slabs of infinite diameter opposite to each other and a perpendicular beam crossing them. That means the (1D-)spatial variable can be identified as the penetration depth and the directional variable $\Omega$ collapses to $\mu$, the cosine of the angle between direction of flight and the beam. In this case our system becomes rotational symmetric.

\subsubsection{Photon transport}
We start with the linear Boltzmann transport equation (also called radiative transfer equation in this context) 
\begin{equation}
\label{eq:transport_photon}
\frac{1}{c} \partial_t \psi(x,\Omega,\epsilon,t) + \Omega \cdot \nabla \psi(x,\Omega,\epsilon,t) =  L_P \psi(x,\Omega,\epsilon,t) + Q(x, \Omega, \epsilon,t),
\end{equation}
where we define the photon scattering operator $L_P$ as
\begin{align}
L_P\psi(x,\Omega,\epsilon,t) &:= \kappa(x) (B(T)-\psi(x,\Omega,\epsilon,t)) \\
&+ \int \limits_{S^2} \sigma(x,\Omega' \cdot \Omega,\epsilon) \cdot \psi(x,\Omega',\epsilon,t) \dd\Omega' - \psi(x,\Omega,\epsilon,t). \nonumber
\end{align}
The energy is computed according to $\epsilon = h \cdot \nu$ with frequency $\nu$ and the Planck constant $h$. Denote the number of photons by $f$.
Additionally, we define the specific intensity 
$\psi(x,\Omega,\epsilon,t) = c\epsilon\, f(x,\Omega,\epsilon,t),$
which has the physical interpretation that $\psi(x,\Omega,\epsilon,t)\cos(\theta)\dd\epsilon \dd\Omega \dd A \dd t$ describes the amount of radiant energy in the interval $(\epsilon,\epsilon+\dd\epsilon)$ traveling in time $\dd t$ through area $\dd A$ into the element of solid angle $\dd\Omega$ around $\Omega$, where $\theta$ is the angle between $\Omega$ and the normal of the area $\dd A$. The specific intensity describes the radiation field inside the medium. We still need to define other parameters: the absorption coefficient, which tells us that a photon, traveling a distance $\dd s$ is absorbed with probability $\kappa(x,\epsilon)\dd s$. Note that absorption and emission share the same coefficient. $\sigma(x,\Omega' \cdot \Omega,\epsilon)$ is the differential scattering cross section describing the probabilty that a photon at point $x$ with energy $\epsilon$ moving in direction $\Omega'$ is scattered to direction $\Omega$.

The radiative transfer equation is a balance law for the conservation of energy. The left hand side is a transport part which describes how photons travel through the medium and the right hand side is a source term accounting for contributions to the radiation field. The integral appears due to the scattering of photons into our beam (in-scattering) and the last term describes the number of photons scattered out of our beam (out-scattering). More specifically, we assume the medium to be in local thermodynamic equilibrium (LTE), which allows us to describe the emission of photons inside the medium by Planck's distribution function. This is an example of a homogeneous and isotropic field, describing the radiation of a perfectly black body in thermodynamic equilibrium at temperature $T$,
\begin{equation}
\psi = B(T) := \frac{2\epsilon^3}{h^2c^2} \left (\exp(\frac{\epsilon}{kT})-1 \right )^{-1}.
\end{equation}
Here $c$ is the speed of light and $k$ is the Boltzmann constant.

Keep in mind though, that the radiative transfer equation is in itself only an approximate description of the propagation of electro magnetic radiation through matter. Two aspects are neglected in this description: firstly, the state of polarization of the electro magnetic field is not taken into account and secondly, photons are treated as particles, neglecting their wave-like behavior. 

\subsection{Method of moments}
\label{sec:methodofmoments}
In this section we investigate the method of moments, which allows us to solve the transfer equation (\ref{eq:transport_photon}), by expanding it into a coupled system of partial differential equations independent of the angular variable, which we will in turn solve numerically. This process can be thought of as averaging over all directions. As we will see, this will lead to a closure problem which we will tackle in section \ref{chap:me_closure}.

A possible motivation for this method is as follows: the transfer equation is a mesoscopic equation, between micro- and macroscopic, describing the exact distribution of all photons in space and time. But in radiation therapy we are in general only interested in macroscopic quantities like the energy density. Hence, it makes sense to try to achieve this transition in scale by integration. Additionally, we can interpret the transfer equation as an infinite system of equations, one for each direction, which we want to replace by a finite number of equations.

From now on, we shall use the following notation 
\begin{definition}
\begin{align}
\langle \cdotp \rangle : = \int \limits_0^\infty \int \limits_{S^2} \cdotp \dd\Omega \dd\epsilon
\end{align}
and introduce the term \textit{moment} mathematically:
we call 
\begin{equation}
\psi^{(i)} := \langle \Omega^i \psi \rangle
\end{equation}  
the \textit{$i^\text{th}$ moment} of $\psi$, with 
\begin{equation}
 \Omega^{0} := 1, \quad \Omega^{1} := \Omega \text{  and} \quad \Omega^{i} := \Omega \otimes \overset{i}{\ldots}  \otimes \Omega.
\end{equation}  
Note that $\Omega^i$ is a tensor of $i^\text{th}$ rank and the integration is done component-wise. Thus $\psi^{(i)}$ is also a tensor of $i^\text{th}$ rank and an element of $\mathbb{R}^{3 \times \overset{i}{\ldots} \times 3 }$.
\end{definition}
\begin{remark}
The zeroth, first and second angular moment of the distribution are called energy density, radiative flux or radiative pressure respectively.
\end{remark}
\begin{remark}
Averaging over all energies leads to so-called grey approximations.
\end{remark}
\begin{definition}
Let us denote by $\tilde{\Omega}^i_k$ the linearly independent entries of $\Omega^i$, where $1 \leq k \leq \frac{i^2}{2}+\frac{3}{2}i+1$.
Now $m(\Omega)$ is defined as the vector of all $\tilde{\Omega}^i_k \, \forall i\leq n \, \forall k$
\begin{equation}
m(\Omega) := (1, \tilde{\Omega}^1_1, \ldots, \tilde{\Omega}^1_{3},\tilde{\Omega}^2_1,\ldots,\tilde{\Omega}^n_{\frac{n^2}{2}+\frac{3}{2}n+1})^T.
\end{equation}
Similarly, we denote by $E$ the vector of the linearly independent entries of the first $n$+1 moments:
\begin{equation}
E := (\langle \psi \rangle, \langle \tilde{\Omega}^1_1 \psi \rangle,\ldots, \langle \tilde{\Omega}^n_{\frac{n^2}{2}+\frac{3}{2}n+1} \psi \rangle)^T.
\end{equation}
\end{definition}
In general, to construct a $n^\text{th}$ order moment system ($n \geq 0$), one calculates all moments of up to order $n$ of every equation, thereby increasing the size of the system. That means multiplying the equation by $m(\Omega)$ and component-wise integration, as defined before. The general $n^\text{th}$ order moment systems for photon transport is
\begin{align}
\label{eq:momentmodel_photons}
\frac{1}{c}\partial_t \langle m(\Omega) \cdot \psi(x,\Omega,\epsilon,t) \rangle + \nabla \cdot \langle \Omega m(\Omega) \cdot \psi(x,\Omega,\epsilon,t) \rangle = \\
\kappa(x) \langle m(\Omega) \cdot (B(T)-\psi(x,\Omega,\epsilon,t))\rangle \nonumber \\ + \langle m(\Omega)(\int \limits_{S^2} \sigma(x,\Omega' \cdot \Omega,\epsilon) \cdot \psi(x,\Omega',\epsilon,t) \dd\Omega' - \psi(x,\Omega,\epsilon,t)) \rangle \nonumber
\end{align}
where the moments of the differential scattering cross section appear in the last term, since
\begin{align}
 \langle \Omega^k(\int \limits_{S^2} \sigma(x,\Omega' \cdot \Omega) \cdot \psi(x,\Omega',\epsilon,t) \dd\Omega'\rangle = \langle \Omega^k \underbrace{2\pi\int\limits_{-1}^1 \mu^k\sigma(x,\mu,\epsilon)\dd\mu}_{= \sigma^{(k)}(x,\epsilon)}\cdot \psi(x,\Omega,\epsilon,t)\rangle.
\end{align}

As can be understood easily, the $n^\text{th}$ order system will always contain the ${(n+1)}^\text{st}$ moment of the distribution, that means first we need to solve the arising closure problem, i.e. find \ $\psi^{(n+1)} := \psi^{(n+1)}(\psi^{(0)},\ldots,\psi^{(n)})$. The way we choose our closure determines the capabilities of our system to model physical situations correctly and should additionally guarantee certain desirable mathematical properties, e.g.\ existence of a solution.

\section{Minimum Entropy Closure}
\label{chap:me_closure}
In this section we will introduce the minimum entropy closure and analyze the resulting system. The idea is to define the highest order moment to be the respective moment of the distribution which minimizes the mathematical entropy of the system while reproducing the (given) lower, $0^\text{th}$ until $n^\text{th}$, order moments. We have to bear in mind that there is a whole family of distributions which fulfill the latter condition. Among those we choose the one which minimizes the entropy of the system, since this is the physically most probable one. The mathematical formulation of this optimization problem is as follows
\begin{align}
\label{eq:mini_prob}
\underset{\psi}{\min} & \quad H(\psi)\\
\text{s.t.} & \quad \langle \psi \cdot m \rangle = E, \nonumber
\end{align}
where $H$ is the entropy, $E$ a vector of prescribed moments, $\psi$ and $m$ as before. After obtaining the minimizer $\psi_{ME}$ we set 
\begin{equation}
\psi^{(n+1)} := \langle \tilde{\Omega}^{n+1} \psi_{ME }\rangle.
\end{equation}
\begin{definition}
The entropy for bosons is \cite{Ros54,Ore55}
\begin{align}
H_B(\psi) = \langle \frac{2k\nu^2}{c^3}((n+1) \ln(n+1) - n \ln(n)) \rangle
\end{align}	 
where n is the occupation number and it holds that
$\psi = \frac{2h\nu^3}{c^2}n.$
\end{definition}

Additionally, we introduce the following
\begin{definition}
$\mathcal{M}_n$ denotes the minimum entropy approximation of $n^\text{th}$ degree, i.e.\ whose solution reproduces the first $n+1$ moments.
\end{definition}

\subsection{Minimum entropy solution}
\label{sec:MinimumEntropySolution}
The minizer for equation (\ref{eq:mini_prob}) is given by \cite{dubrofeu99}
\begin{align}
\label{eq:psi_ME}
	\psi_{ME}(\alpha)= \frac{2h\nu^3}{c^2} (\exp(\frac{h\nu}{kT} \alpha \cdot m)-1)^{-1},
\end{align}
where $\alpha \in \mathbb{R}^n$ is the vector of Lagrange multipliers.

The non-negativity of the solution shows a clear advantage of the minimum entropy model compared to the spherical harmonics for example, since in the latter unphysical, negative distributions can occur. It is reasonable to simplify the frequency-integral by applying the Stefan-Boltzmann law:
\begin{align}
\label{eq:Stefan-Boltzmann-law}
\langle \psi_{ME}(\alpha) \rangle &= \int \limits_0^\infty \int \limits_{S^2} \frac{2h\nu^3}{c^2} (\exp(\frac{h\nu}{kT} \alpha \cdot m)-1)^{-1} \dd\Omega \dd\nu
&= \sigma_{\text{stefan}} \int \limits_{S^2} \frac{T^4}{(\alpha \cdot m)^4} \dd\Omega,
\end{align}
where $\sigma_{\text{stefan}}$ is the Stefan-Boltzmann constant.

We write our minimizer $\psi_{ME}(\alpha)$ in dependency of the Lagrange multipliers, instead of the spatial, angular, energy and time variables, as in $\psi(x,\Omega,\epsilon,t)$. The Lagrange multipliers have to be determined from the set of constraints $\langle \psi \cdot m \rangle = E$ and thus depend on the moments $\psi^{(i)}$ $0 \leq i \leq n$ and therefore on $(x,\Omega,\epsilon,t)$ just as before.

Let us state important properties: obviously the underlying distribution is always non-negative. Additionally it can be shown, cf.\ \cite{dubrofeu99}, that the moment system closed by the minimum entropy closure is symmetrizable hyperbolic.

Now we take a closer look at the constraints that restrict the distribution $\psi$ and its moments or the normalized moments respectively.
\begin{remark}
From now on, we shall only consider the one-dimensional case, that means our model assumes slab-geometry. Effectively $\Omega$ is replaced by $\mu$ and hence all moments become one-dimensional quantities.
\end{remark}
With this assumption it becomes easy to see that the absolute value of $i^\text{th}$ moment decreases as i increases, i.e.\ 
\begin{align}
\label{eq:psisgetsmaller}
\left | \psi^{(i+1)} \right | \leq \left | \psi^{(i)} \right | \quad \forall i \in \mathbb{N},
\end{align}
which tells us that the normalized flux is limited, i.e.\ $\left | \frac{\psi^{(1)}}{\psi^{(0)}} \right |\leq 1$ which guarantees that the speed of propagation in our model is limited by the speed of light, a fact that does not hold in diffusion models. This follows from a study of the characteristic velocities of the system, i.e.\ the speed of information propagation.

Now that we calculated the minimum entropy solution, the closure is defined by
\begin{align}
\psi^{(n+1)} := \langle \mu^{n+1} \psi_{ME }\rangle.
\end{align}
As mentioned before, the Lagrange multipliers have to be determined from the set of constraints $\langle \psi \cdot m \rangle = E$, i.e.\ it is in general not possible to express $\psi^{(n+1)}$ explicitly in terms of the lower order moments. Therefore we introduce the following notation:
\begin{align}
\psi^{(n+1)} = \chi(\frac{\psi^{(1)}}{\psi^{(0)}},\ldots,\frac{\psi^{(n)}}{\psi^{(0)}}) \cdot \psi^{(0)},
\end{align}
where $\chi$ is called Eddington factor. One is allowed to work with normalized quantities since then the Eddington factor no longer explicitly depends on the first Lagrange multiplier, see \cite{WriFraKla08}. A further aspect to notice is that the Eddington factor is an odd/even function on the odd/even moments, which follows from the symmetry of the moments.

\subsection{$\mathcal{M}_1$}
\label{sec:M1}
It is not possible to derive an explicit formula of the Eddington factor for arbitrary dimension $n$. However, there is an analytical solution for $\mathcal{M}_1$, see \cite{dubrofeu99} for a proof:
\begin{equation}
\chi(\frac{\psi^{(1)}}{\psi^{(0)}}) = \frac{3+4(\frac{\psi^{(1)}}{\psi^{(0)}})^2}{5+2\sqrt{4-3(\frac{\psi^{(1)}}{\psi^{(0)}})^2}}
\label{eq:shi_M1}
\end{equation}

\begin{figure}[ht]
  \centering
  \includegraphics[width=0.66\textwidth]{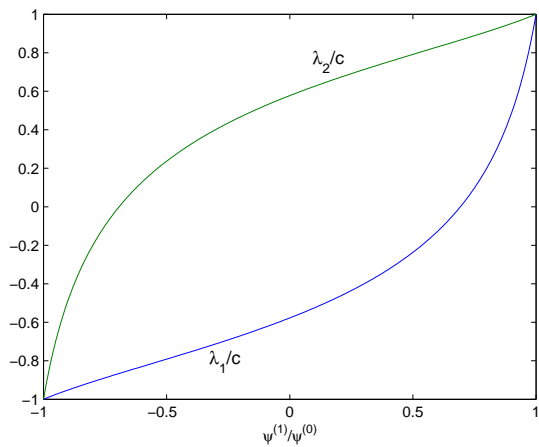}
  \caption{Characteristic velocities of the $\mathcal{M}_1$ system}
  \label{fig:shi_n2_EV}
\end{figure}

See figure \ref{fig:shi_n2_EV} for a plot of the (analytically expressible) eigenvalues of the $\mathcal{M}_1$ system. Let us observe that when the absolute value of the first normalized moment approaches 1, the distribution has to become more and more peaked around $\pm1$. That means the direction of travel for all particles has to be in a smaller and smaller cone around the angle $\arccos(\pm1)$. In the limit the distribution becomes a Dirac, $\psi = \delta(\mu\pm1)$. 

\subsection{$\mathcal{M}_2$}
\label{sec:M2}
In case of $\mathcal{M}_2$, it is also possible to calculate an (implicit) formula of the Eddington factor. The formula is very lengthy and includes certain integral expressions and is therefore omitted at this point, see appendix \ref{sec:appendix_shi_boson} in the appendix. To speed up the evaluation we use a rational function with 15 parameters, fitted using a least-squares method to the aforementioned formula, see appendix \ref{sec:appendix_shi_fit} and \cite{ranvinc08}. Additionally we will use the fit as a means of verification for our numerical computations in the next section.

By elemental calculus one can show that
\begin{equation}
\label{eq:psi_geq_condition}
\frac{\psi^{(2)}}{\psi^{(0)}} \geq \left(\frac{\psi^{(1)}}{\psi^{(0)}}\right)^{2}
\end{equation} 
holds. This in turn further confines the domain on which the Eddington factor is defined. Corto and Fialkow derived in \cite{corfia91} general realization conditions of truncated moment problems, but in the present case properties can be derived by direct calculations. We want to study what happens when the normalized moments approach the boundary of their admissible domain and establish boundary curves, which yield the extreme values of the Eddington factor.

\begin{lemma}
Assume that inequality (\ref{eq:psi_geq_condition}) is sharp, i.e.\ we are at one boundary of the domain of the Eddington factor, the following holds:
\begin{itemize}
\item[(i)] The underlying distribution is a Dirac delta distribution
\begin{equation*}
\psi = \delta(\mu-a) \qquad a \in [-1,1].
\end{equation*}
\item[(ii)]
The normalized moments, are of the form
\begin{equation}
\label{eq:boundary1}
\frac{\psi^{(i)}}{\psi^{(0)}} = a^i \qquad a \in [-1,1] \quad \forall i \leq n+1.
\end{equation}
\item [(iii)]
\begin{align}
\label{eq:boundary2}
\frac{\psi^{(2)}}{\psi^{(0)}} = 1 \quad \text{and} \quad \frac{\psi^{(3)}}{\psi^{(0)}} = \frac{\psi^{(1)}}{\psi^{(0)}} = a \qquad a \in [-1, 1]
\end{align}
represents a curve and together with equation (\ref{eq:boundary1}) defines the admissible set of moments for $\mathcal{M}_2$.
\end{itemize}
\end{lemma}
\begin{proof}
(i) We start by assuming that
\begin{equation}
\frac{\psi^{(i)}}{\psi^{(0)}} = \left(\frac{\psi^{(1)}}{\psi^{(0)}}\right)^i
\end{equation}
holds for some even $i \geq 2$. This means that
\begin{equation}
\int \limits_{-1}^1 \mu^i \psi =  \left (\int \limits_{-1}^1 \mu \psi \right )^i,
\end{equation}
but still $|\mu|\leq 1$. This equality can only hold for $\psi = \delta(\mu - a)$ for some $a \in [-1,1]$, as can be seen by using the convolution property of the Dirac delta distribution
\begin{equation}
\int \limits_{-\infty}^\infty f(x) \delta(x-a)dx = f(a).
\end{equation}
(ii) Obvious, using (i) and the convolution property again\\
(iii) We can deduce from equation (\ref{eq:psi_geq_condition}) that when the first normalized moment approaches $\pm 1$, the second and third normalized moments approach 1 or $\pm 1$ resp.\ where the limit for the third normalized moment follows from (\ref{eq:psi_geq_condition}). More specifically
\begin{equation}
\frac{\psi^{(1)}}{\psi^{(0)}} =  \pm 1 \quad \text{implies} \quad \frac{\psi^{(2)}}{\psi^{(0)}} = 1, \frac{\psi^{(3)}}{\psi^{(0)}} = \pm 1.
\end{equation}
The physical interpretation for either of these two extreme cases is a beam parallel or anti parallel to the main axis of interest, i.e. the intensity becomes $\delta(\mu - 1)$ or $\delta(\mu + 1)$ respectively. Now considering the case of 2 beams and applying the superposition principle, we expect the distribution to be a linear combination of these two Diracs, i.e. for $\frac{\psi^{(2)}}{\psi^{(0)}} = 1$ we have
\begin{equation}
\psi = (\frac{1}{2} - \frac{a}{2}) \cdot \delta(\mu-1) + (\frac{1}{2} + \frac{a}{2}) \cdot \delta(\mu+1), a \in [-1,1]
\end{equation}
and by calculation of moments we get
\begin{equation}
\frac{\psi^{(3)}}{\psi^{(0)}} = \frac{\psi^{(1)}}{\psi^{(0)}}
\end{equation}
from which we deduce the form of the second boundary curve
\begin{equation}
\frac{\psi^{(2)}}{\psi^{(0)}} = 1 \quad \text{and} \quad \frac{\psi^{(3)}}{\psi^{(0)}} = \frac{\psi^{(1)}}{\psi^{(0)}} = a \qquad a \in [-1, 1].
\end{equation}
\end{proof}

The analysis of the boundary helps us to understand the capabilities of the model to simulate certain scenarios, e.g.\ the usage of multiple beams and assists us in verifying our numerical computations later on. It will be extended to higher orders in section \ref{chap:properties_Mn}, though the upper bound on the propagation speed remains the same. For a more detailed study of different models, but in the electron case, see also \cite{frahenkla06}.

\section{Computation of the Eddington Factor}
\label{chap:computation_eddi}
As mentioned earlier it is not possible to calculate the Eddington factor for a general $\mathcal{M}_n$ closure analytically. This section is devoted to the necessary numerical calculations. Here we present a Monte-Carlo-type ansatz. First, we provide Lagrange-multipliers and calculate the Eddington factor as the highest order normalized moment directly. This is possible since we know the general form of the minimizer, 
\begin{equation}
\psi_{\text{ME}} = \underbrace{T^4 \sigma_{stefan}}_{=\text{const}} \int \limits_{-1}^1 \frac{1}{(\alpha^{(0)} + \mu  \alpha^{(1)} + \ldots + \mu^n \alpha^{(n)})^4} \dd\mu.
\end{equation}
To generate valid moments, observe that the polynomial $(\alpha^{(0)} + \mu  \alpha^{(1)} + \ldots + \mu^n \alpha^{(n)})^4$, appearing in the denominator, must not have real roots within the interval $[-1,1]$, whereas complex roots are allowed. The exception to this is the boundary curve. As we know from section \ref{sec:M2}, the distribution tends toward a Dirac, when the two normalized moment approach each other, i.e. 
\begin{align}
\frac{\psi^{(i)}}{\psi^{(0)}} \rightarrow \left(\frac{\psi^{(j)}}{\psi^{(0)}}\right)  \Rightarrow \psi \rightarrow \delta(r-\mu),
\end{align}

where $r$ is the unique root of the polynomial in the denominator and it holds that $r \in [-1,1]$. Technically, in our numerical simulations we circumvent the division by zero by calculating the denominator first and if it equals zero, setting the integral to 1. The last choice is no approximation but the analytical solution, since $\int_{-\infty}^\infty \delta(x) \dd x = 1$. It is important to note that the root is only unique for this boundary, while for higher orders and other boundaries the distribution becomes a linear combination of various Diracs. Hence multiple roots in the interval $[-1,1]$ must be allowed.

Now we randomly, hence titled Monte-Carlo-type approach, choose (valid) \\Lagrange-multipliers $\alpha$ and linearly interpolate the result $\frac{\psi^{(n+1)}}{\psi^{(0)}}$ in the space of normalized moments $\left ( \frac{\psi^{(1)}}{\psi^{(0)}}, \frac{\psi^{(2)}}{\psi^{(0)}}, \ldots, \frac{\psi^{(n)}}{\psi^{(0)}} \right )$. We utilize randomness in this ansatz, since it is at first glance not clear at all how to transverse the space of normalized moments by perturbation of the Lagrange multipliers. For a non-random algorithm, see \cite{WriFraKla08}. So we simply create enough data points to sufficiently fill our domain. Note that figures are shown for 1000 data points for demonstration purposes, while later computations rely on Eddington factors interpolated from far more data points. Interpolation is done using MATLAB, see \cite{sand87} for details of the interpolation algorithm.

See figure \ref{fig:shi_n2_mc} for a comparison between the analytical and the numerical solution for $\mathcal{M}_1$, and figure \ref{fig:diff_shi_ph_el_3_mc} for the $\mathcal{M}_2$ closure and its comparison the rational fit.

\begin{figure}[htb]
  \centering
  \includegraphics[width=0.5\textwidth]{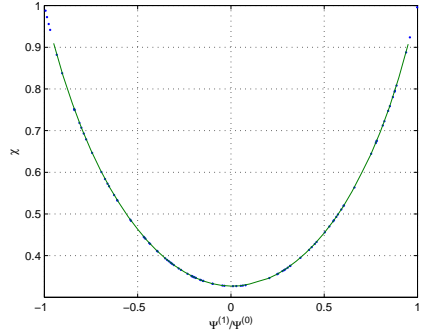}  
  \caption{$\mathcal{M}_1$ Eddington factor}
  \label{fig:shi_n2_mc}
\end{figure}

\begin{figure}[htbp]
  \centering
  \includegraphics[width=0.49\textwidth]{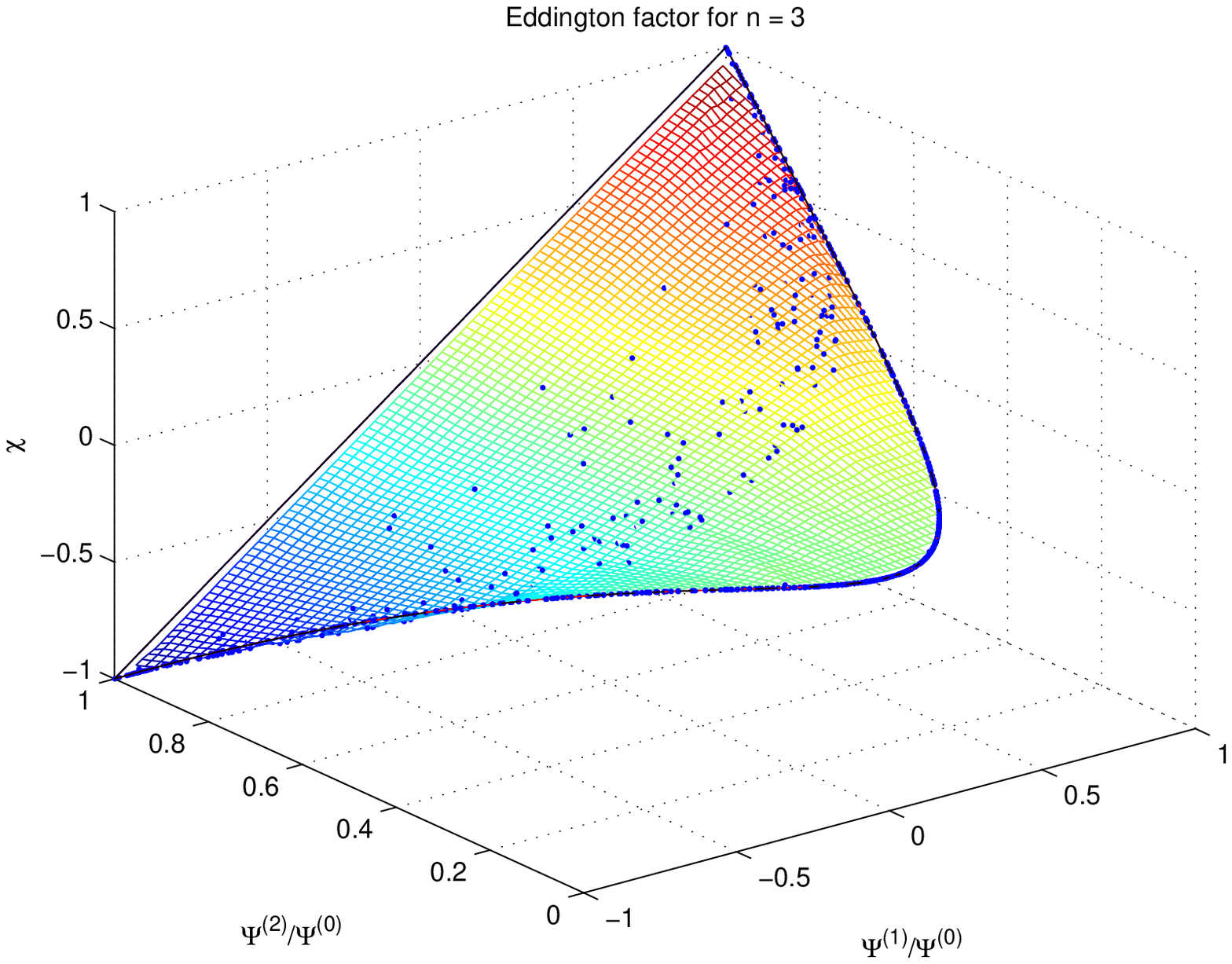}
  \includegraphics[width=0.49\textwidth]{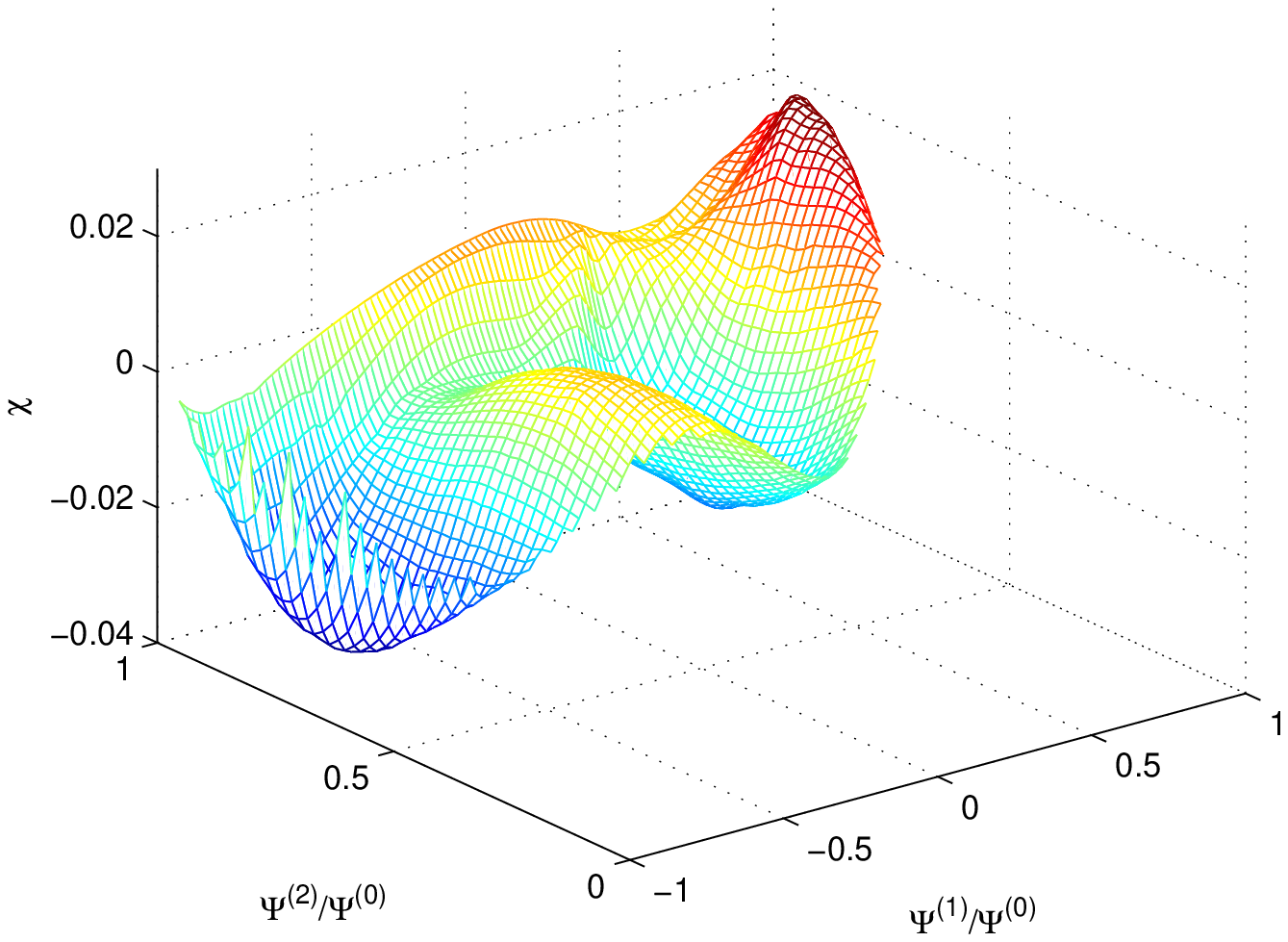}
  \caption{$\mathcal{M}_2$ Eddington factor with 1000 data points and difference to the rational fit}
  \label{fig:diff_shi_ph_el_3_mc}
\end{figure}

\section{Properties of the $\mathcal{M}_n$ Systems}
\label{chap:properties_Mn}
Now, with some important properties of the minimum entropy solution, -system and Eddington factor under our belt, we shall further investigate the properties which can be derived from the hyperbolicity of the system. In order to do so, we rewrite the system into the following form:
\begin{equation}
U_t + F(U)_x = C(U,T)
\label{eq:hyperbolic_form}
\end{equation}
with
\begin{equation}
U := \frac{1}{c} \left ( \begin{matrix} \psi^{(0)} \\ \psi^{(1)}  \\ \psi^{(2)}  \\ \vdots  \end{matrix} \right ), \qquad 
F(U) := \left ( \begin{matrix} \psi^{(1)} \\ \psi^{(2)}  \\ \psi^{(3)}  \\\vdots \end{matrix} \right ) \quad \text{and} \quad
C(U,T) := \left ( \begin{matrix} \kappa(4\pi B(T)-\psi^{(0)}) \\ -(\kappa+\sigma^{(1)})\psi^{(1)} \\ -(\kappa+\sigma^{(2)})\psi^{(2)} \\ \vdots \end{matrix} \right )
\end{equation}

\subsection{$\mathcal{M}_2$}
\begin{figure}[t!]
  \centering 
  \includegraphics[width=0.49\textwidth]{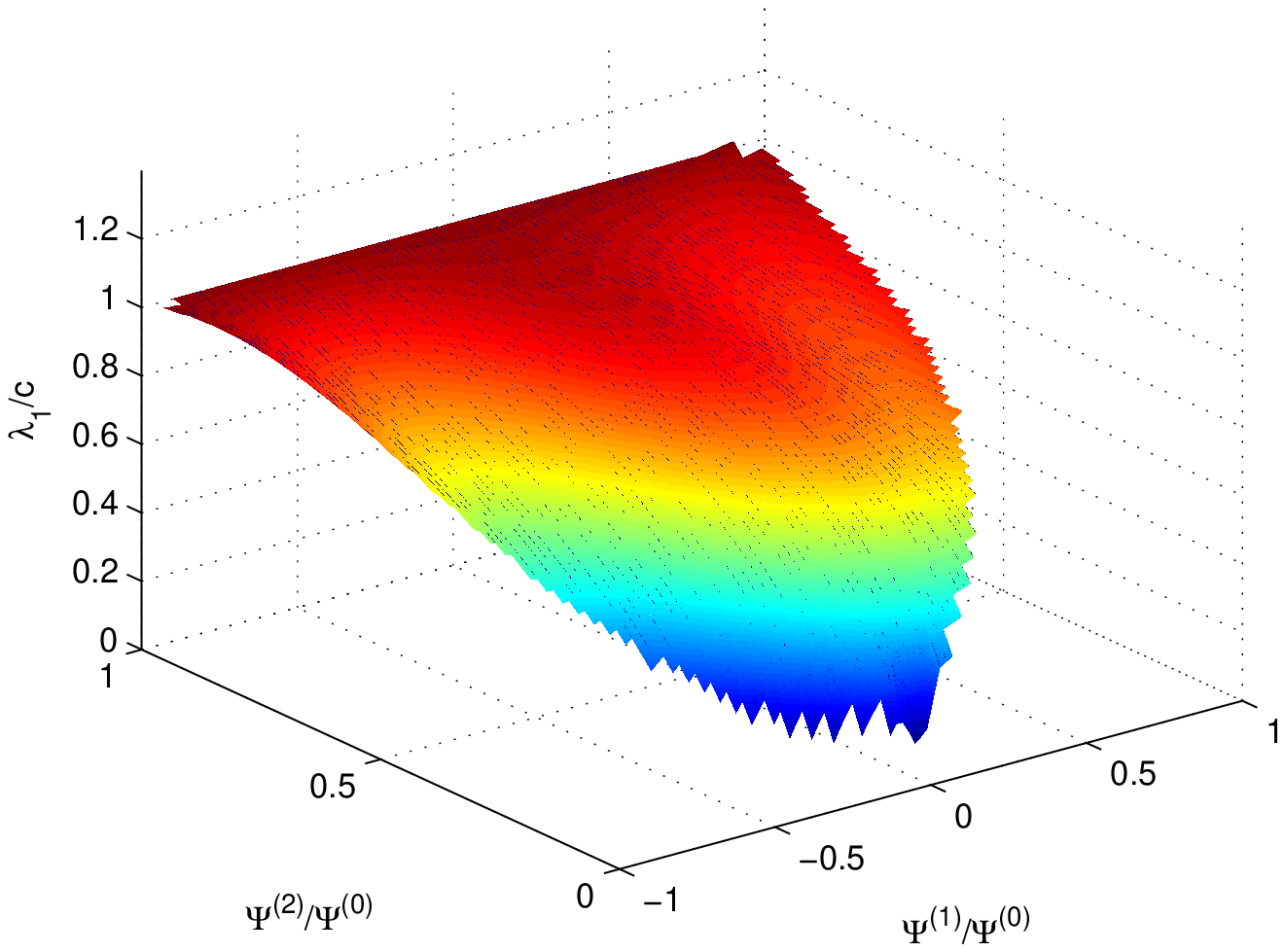}  
  \includegraphics[width=0.49\textwidth]{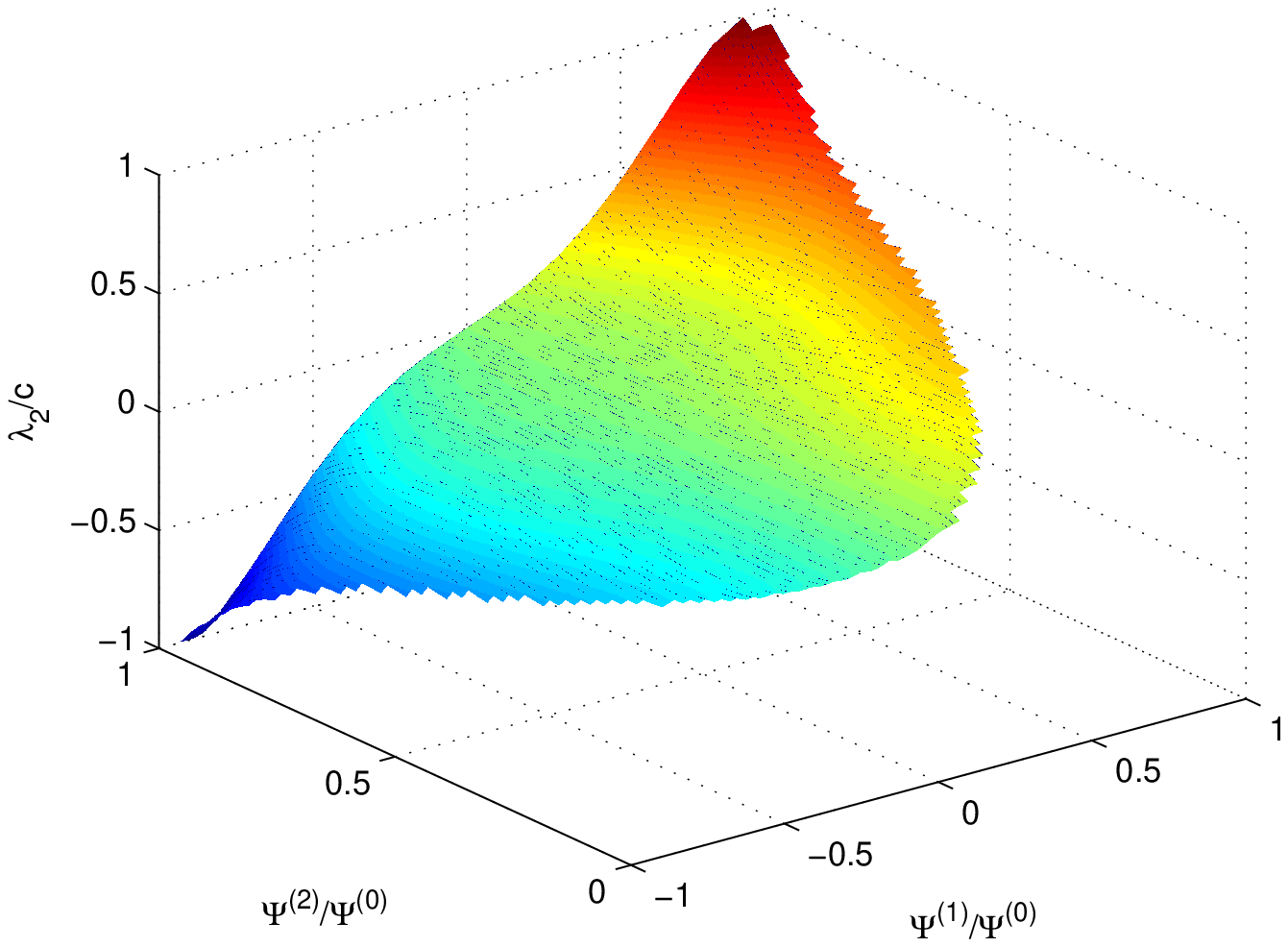}  
  \includegraphics[width=0.5\textwidth]{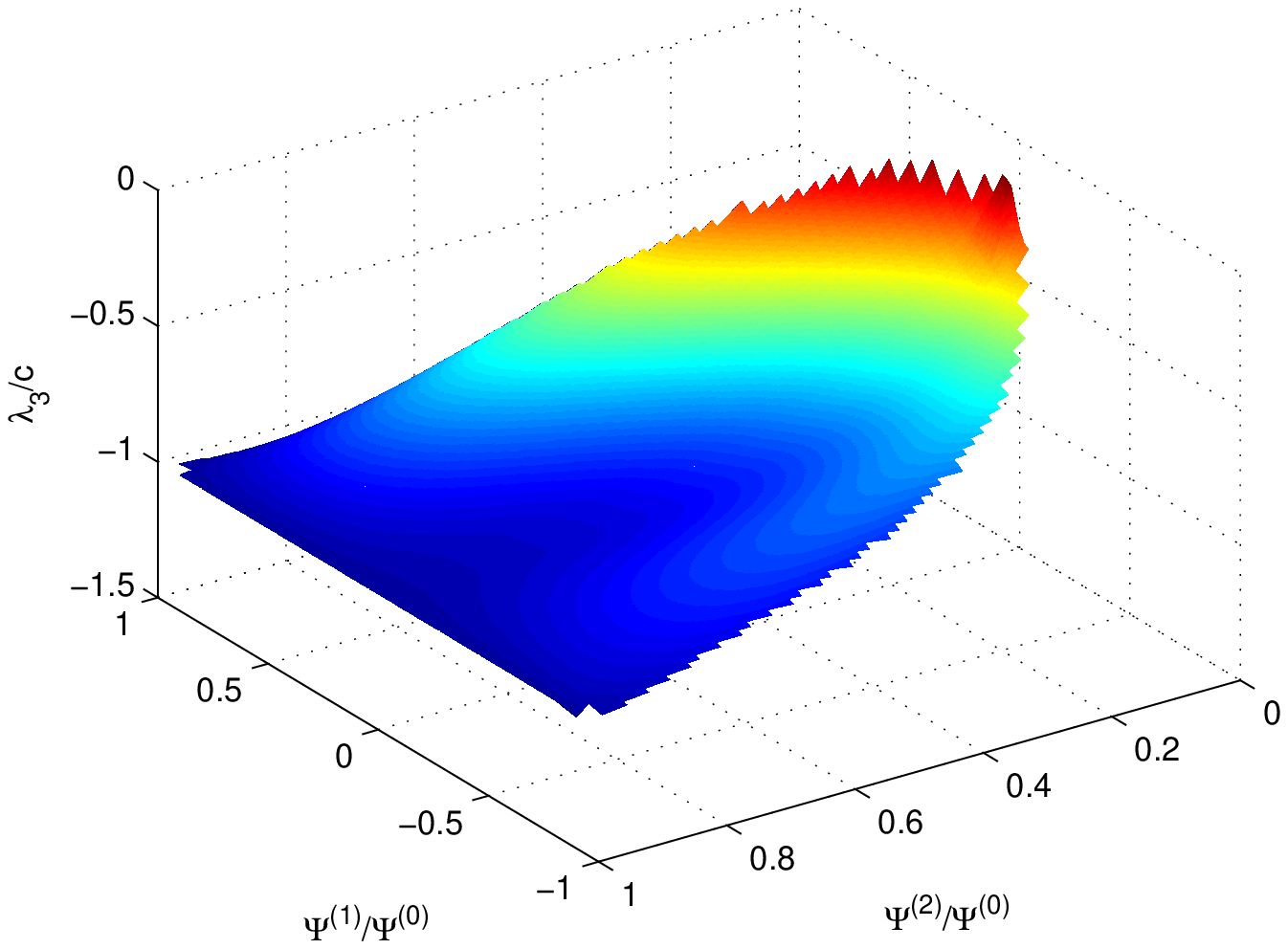}  
  \caption{Characteristic velocities of the $\mathcal{M}_2$ system}
  \label{fig:EV_n3}
\end{figure}
First off we want to calculate the characteristic velocities with which the system propagates information. They will also allow us to gain insight into how the system behaves numerically. The Jacobian matrix J of the system (\ref{eq:hyperbolic_form}) in this case, is given by:
\begin{align}
J(U) &= \frac{\partial F(U)}{\partial U} \nonumber \\
&= c \left ( \begin{matrix} 0 & 1 & 0 \\ 0 & 0 & 1 \\ \frac{\partial \psi^{(3)}}{\partial \psi^{(0)}} & \frac{\partial \psi^{(3)}}{\partial \psi^{(1)}} & \frac{\partial \psi^{(3)}}{\partial \psi^{(2)}} \end{matrix} \right )
\end{align}
where we use
\begin{align}
\frac{\partial \psi^{(3)}}{\partial \psi^{(0)}} = \frac{\partial \chi \psi^{(0)}}{\partial \psi^{(0)}} = \chi - \left ( \frac{\partial \chi}{\partial \left ( \frac{\psi^{(2)}}{\psi^{(0)}} \right )} \frac{\psi^{(2)}}{\psi^{(0)}} +  \frac{\partial \chi}{\partial \left ( \frac{\psi^{(1)}}{\psi^{(0)}} \right )} \frac{\psi^{(1)}}{\psi^{(0)}}  \right ).
\end{align}
Since we do not have an analytical expression for the Eddington factor, we use symmetric differences to approximate the derivatives.
Additionally, for $i \in \{1,2\}$
\begin{align}
\frac{\partial \psi^{(3)}}{\partial \psi^{(i)}} = \frac{\partial( \chi \psi^{(0)})}{\partial \psi^{(i)}} = \psi^{(0)} \cdot \frac{\partial \chi}{\partial \left ( \frac{\psi^{(i)}}{\psi^{(0)}}  \right )} \frac{\partial \left ( \frac{\psi^{(i)}}{\psi^{(0)}}  \right )}{\partial \psi^{(i)}} = \frac{\partial \chi}{\partial \left ( \frac{\psi^{(i)}}{\psi^{(0)}}  \right )}. \nonumber 
\end{align}
Figure \ref{fig:EV_n3} shows the characteristic velocities scaled by the speed of light. As expected all eigenvalues are real, since the system is hyperbolic and velocity of the propagation approaches the speed of light as the normalized moment $\frac{\psi^{(2)}}{\psi^{(0)}}$ approaches 1, a fact that other models are not capable of simulating. This scenario corresponds to the case where all particles are moving in the positive or negative direction of the x-axis, depending on the sign of the characteristic velocity. Note especially the symmetry of the eigenvalues.
 
As mentioned before, it is of great interest to see whether the model can handle radiative non-equilibria. Therefore we study how the radiative intensity behaves while the angular anisotropy increases. Figure \ref{fig:intensity_n3} shows the intensity distribution on a grid of normalized moments for different values of $\mu$ and the maximum intensity for that angle on a logarithmic scale. As would be exptected, the intensity becomes more and more peaked for $|\mu| \rightarrow 1$. and is isotropic for $\mu = 0$.
Figure \ref{fig:intensity_polar_n3} shows the intensity depending on $\arccos(\mu)$ for fixed normalized moments in polar coordinates. The distribution is highly forward peaked for $\frac{\psi^{(2)}}{\psi^{(0)}} \rightarrow 1$ and isotropic for $\frac{\psi^{(2)}}{\psi^{(0)}} = \frac{1}{3}$, as stated in section \ref{sec:methodofmoments}, while the direction of transportation is perpendicular to the x-axis when both normalized moments approach zero which corresponds to the second eigenvalue.

\begin{figure}[htbp]%
\subfloat[][$\mu = 0$]{%
\includegraphics[width=0.5\textwidth]{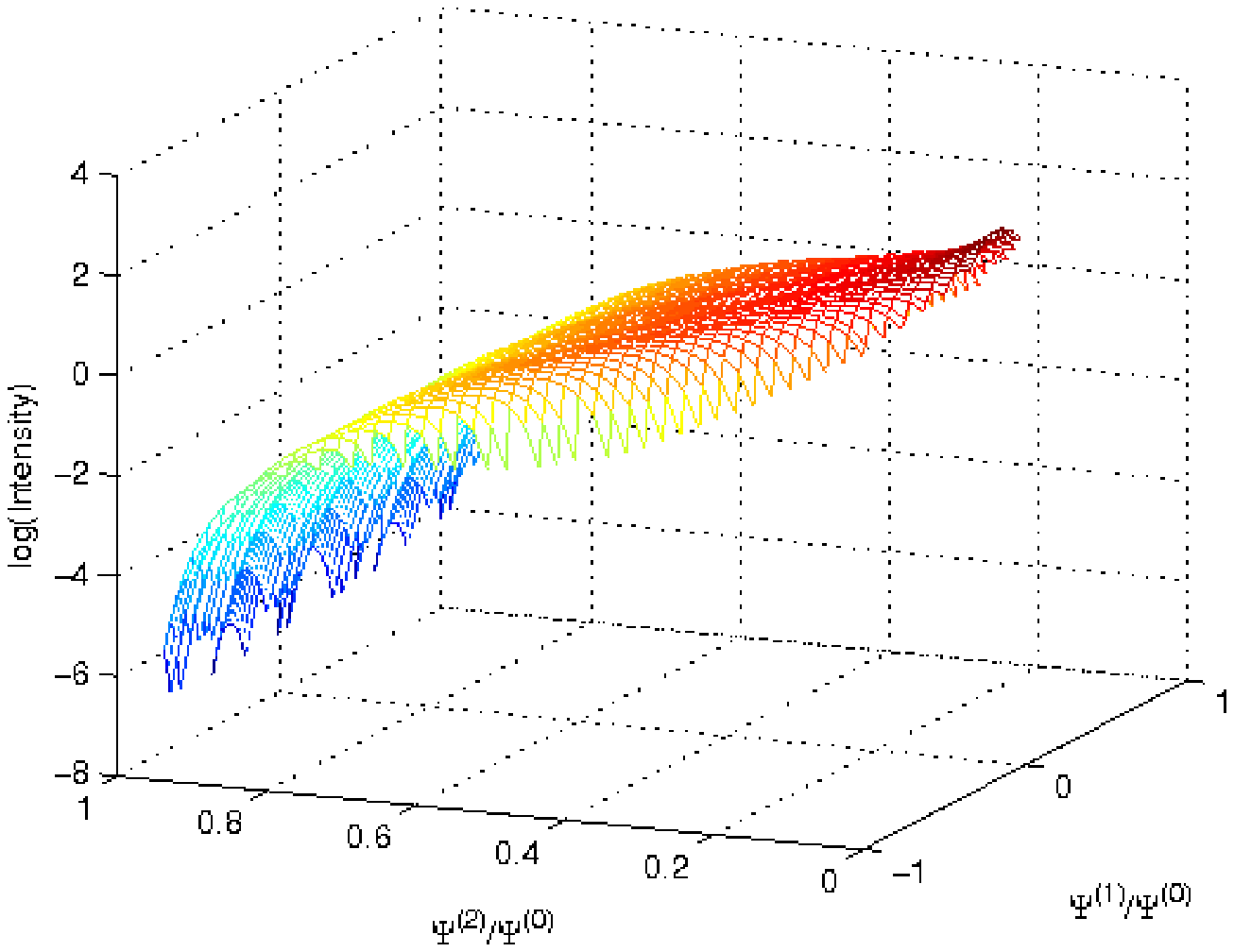}}%
\subfloat[][$\mu = 0.5$]{%
\includegraphics[width=0.5\textwidth]{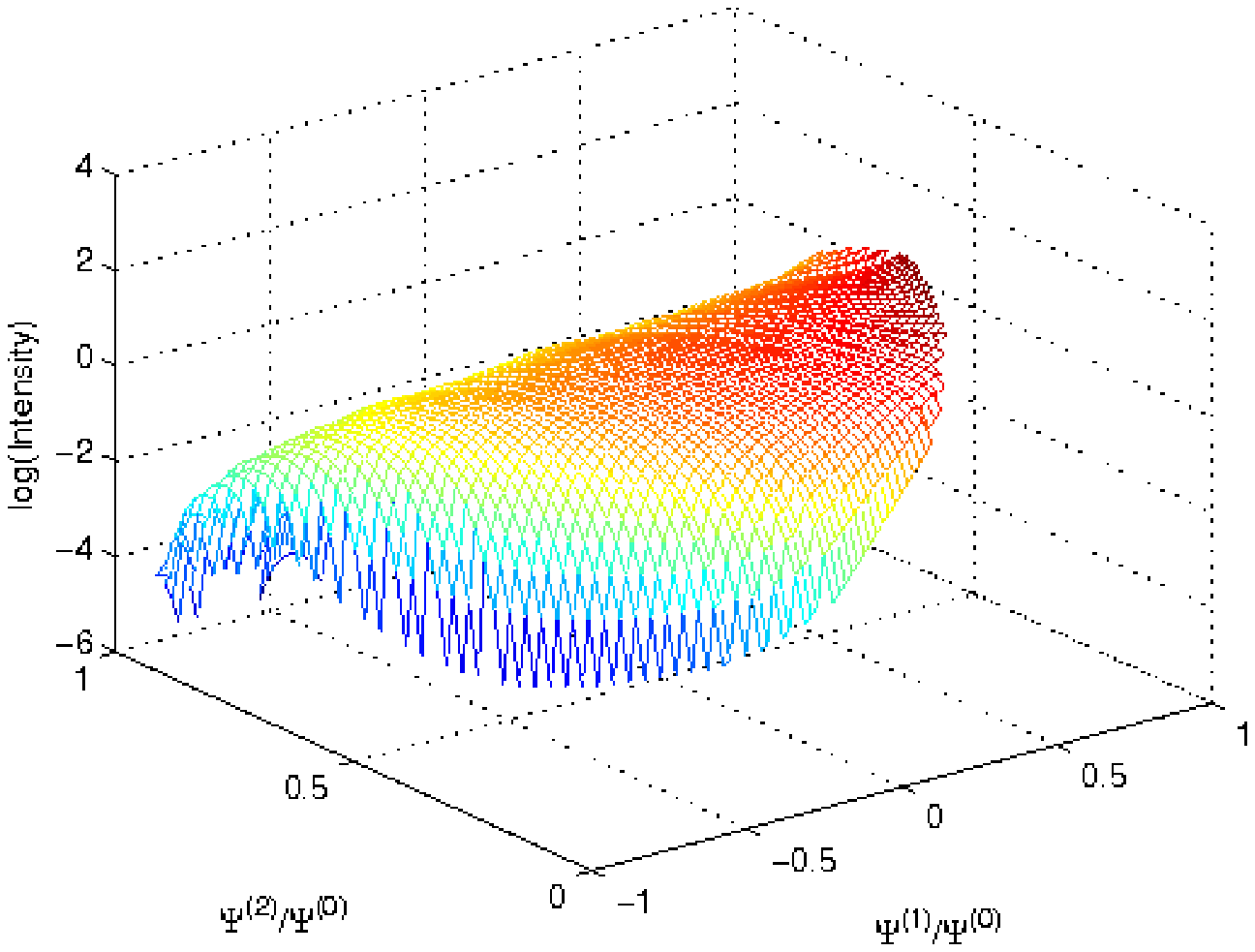}}\\
\centering
\subfloat[][$\mu \rightarrow 1$]{%
\includegraphics[width=0.5\textwidth]{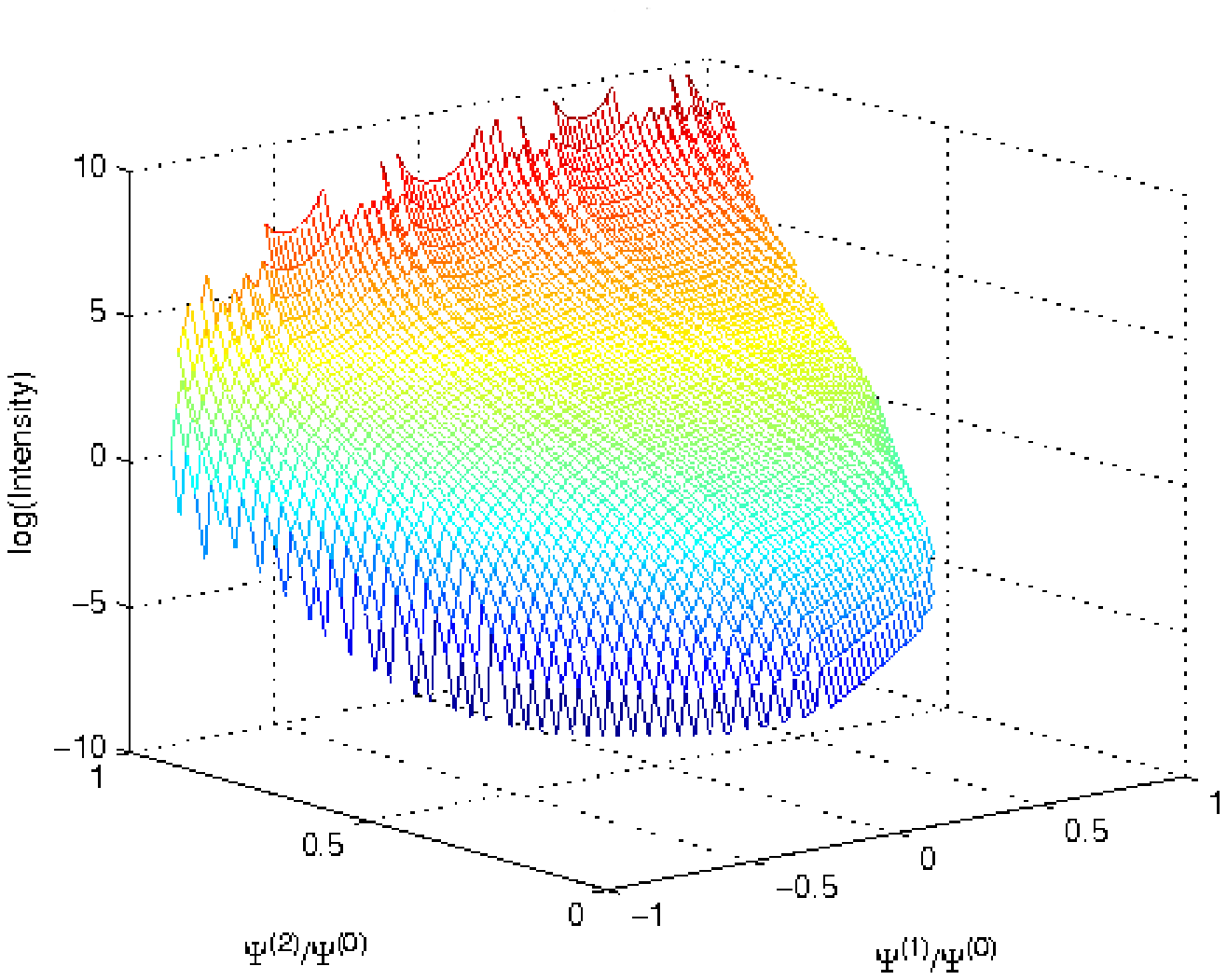}}
\caption{Intensity distribution of the $\mathcal{M}_2$ system}
\label{fig:intensity_n3}%
\end{figure}

\begin{figure}[htbp]%
\subfloat[][$\frac{\psi^{(2)}}{\psi^{(0)}} = 0.14$]{%
\includegraphics[width=0.5\textwidth]{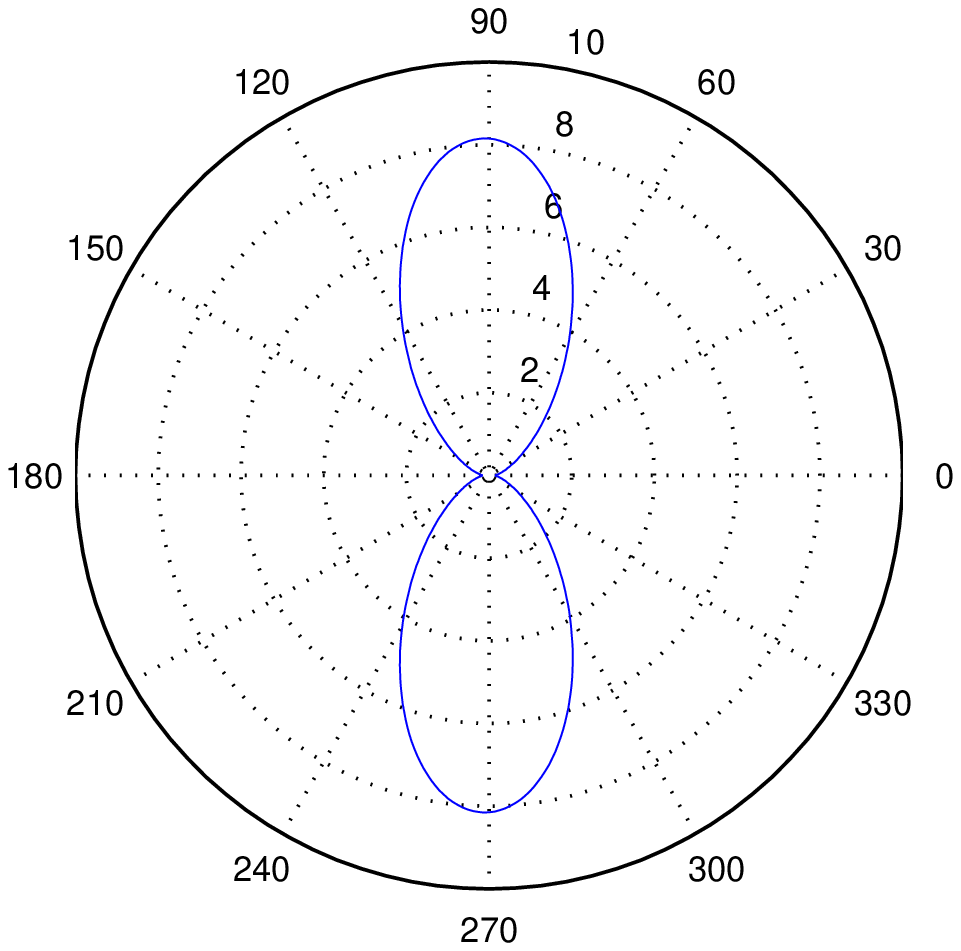}}
\subfloat[][$\frac{\psi^{(2)}}{\psi^{(0)}} = 0.28$]{%
\includegraphics[width=0.5\textwidth]{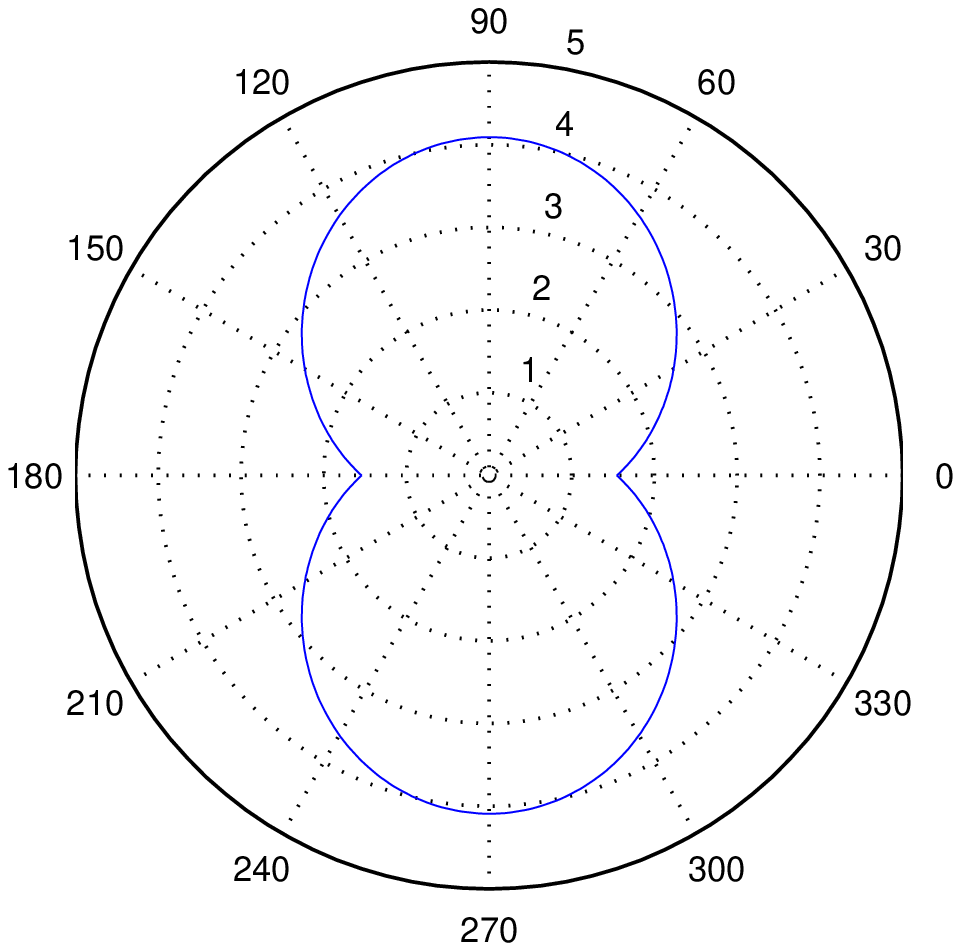}}\\
\subfloat[][$\frac{\psi^{(2)}}{\psi^{(0)}} = 0.42$]{%
\includegraphics[width=0.5\textwidth]{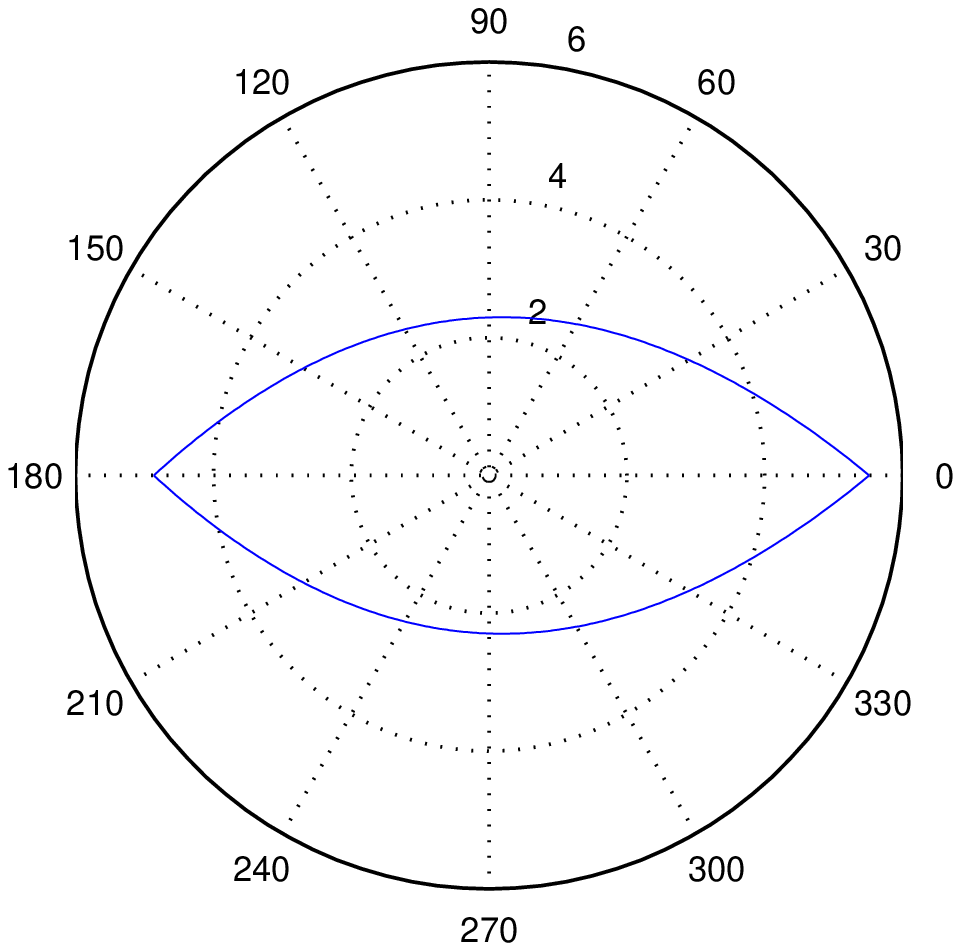}}
\subfloat[][$\frac{\psi^{(2)}}{\psi^{(0)}} = 0.95$]{%
\includegraphics[width=0.5\textwidth]{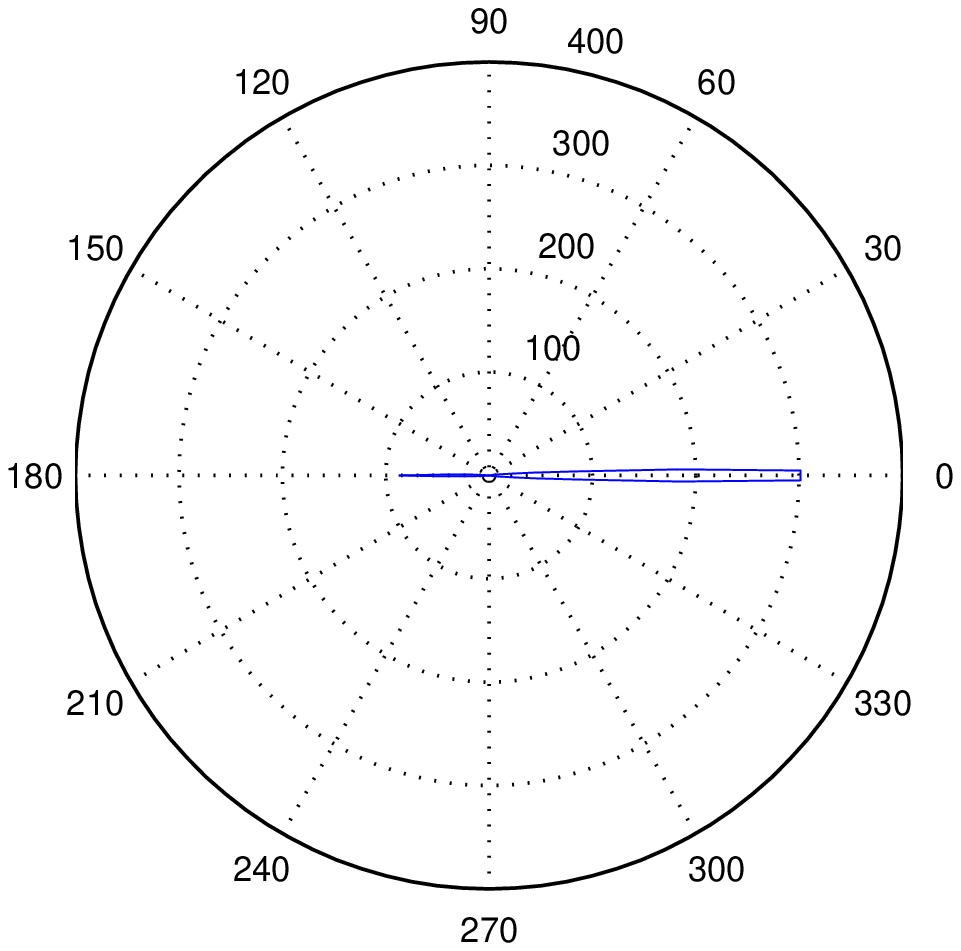}}
\caption{Intensity distribution of the $\mathcal{M}_2$ system; $\frac{\psi^{(1)}}{\psi^{(0)}} = 0.01$ fixed, angles are $\arccos(\mu)$}
\label{fig:intensity_polar_n3}%
\end{figure}

\section{Numerical Results}
\label{chap:num_results}
In this section we discuss numerical simulations using the minimum entropy approximation and compare it to other models of radiative transfer. Firstly the method of solving those systems will be explained, before we study different scenarios, i.e.\ test our simulations with different parameters, boundary or initial conditions. A high order discrete ordinates approximation, cf.\ \cite{kokrwivis95}, is used as a benchmark in the first test case, while in the second scenario we use an analytical solution generated by Su and Olson in \cite{suols95}. Note that we always assume slab-geometry in our simulations, and a 1-dimensional domain. Also, we identify the spatial variable $x$ as the penetration depth of the beam. 
We test minimum entropy approxmations of first and second order, denoted by $\mathcal{M}_1$ and $\mathcal{M}_2$, against the benchmark. Additionally, we add the spherical harmonics methods of order 1 and 3, denoted by $P_1$ and $P_3$ respectively. The $P_N$ approach is one of the oldest approximate methods for radiative transfer and was first introduced by Eddington in 1926, see \cite{eddi26}.
\subsection{Method}
As mentioned earlier, the first order system of partial differential equations resulting from the minimum entropy method is hyperbolic.  We choose to use a kinetic scheme \cite{talper97}, which is a finite volume method, to achieve a spatial discretization in the equations. Here, the fluxes over the cell boundaries are computed, which allows us to prescribe boundary conditions easily. To derive these fluxes, incorporating the closure, we use half-moments, i.e.\
\begin{align}
\psi^{(i)}_+ := \langle \mu^i \cdot \psi \rangle_+ := \int \limits_0^1 \mu^i\cdot \psi \dd\mu \qquad
\psi^{(i)}_- := \langle \mu^i \cdot \psi \rangle_- := \int \limits_{-1}^0 \mu^i\cdot \psi \dd\mu 
\end{align}
where '+' and '-' denote the sign of $\mu$ and hence the direction of transport. To calculate these half-moments, only minor alterations have to be made to the algorithm introduced in section \ref{chap:computation_eddi}. Again, computations (and interpolations) are done on a grid of normalized moments, as we already have the Lagrange multipliers at hand. However, actually only one half-moment has to be computed since we can make use of the symmetry in the moments. 

After the discretizations, we end up with a system of ordinary differential equations, which in turn is being solved by an adaptive Runge-Kutta method.

\subsection{Simulations}
\label{sec:Simulations}
\subsubsection{First test-case}
We set the initial temperature inside the medium to be almost zero and let two rays of light enter at both sides of the domain. The spatial domain is $x \in [0,1]$. We consider a scattering coefficient of $\sigma = 0.01$ and an absorption coefficient of $\kappa = 2.5$. At the boundary we prescribe the moment of the distribution according to $\psi^{(0)} = \sigma_\text{stephan} T^4$ with $T=1000K$.

The test is designed \cite{BruHol01} to show the drawback of the otherwise accurate $\mathcal{M}_1$ model and how $\mathcal{M}_2$ overcomes this problem.

In the figures we depict the $0^\text{th}$ moment of the distribution $\psi$, i.e.\ the radiative energy $E$, depending on $x$. See figure \ref{fig:sim_1} for the results. Note that the $\mathcal{M}_1$ approximation produces an unphysical shock and is qualitatively wrong. The $\mathcal{M}_2$ model fixes this but does not have the same accuracy as the spherical harmonics methods $P_1$ and $P_3$.

\begin{figure}[htp!]%
\includegraphics[width=1\textwidth]{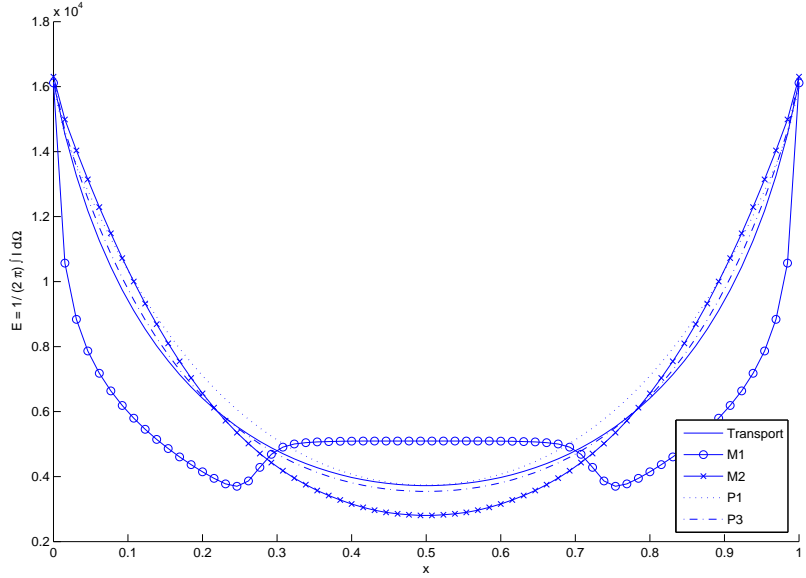}%
\caption{Energy distribution for two entering rays; $\sigma = 0.01$, $\kappa = 2.5$, initial temperature $T_0 \approx 0K$, boundary temperatures $T_l = T_r =1000K$.}
\label{fig:sim_1}%
\end{figure}

\subsubsection{Second test-case}
Here the medium is infinite, $x\in\mathbb{R}$. Initially the medium is at temperature $T=0K$ for $t=0$. A uniform source is applied according to
\begin{align}
 Q = \begin{cases} 1 \quad -x_0 \leq x \leq x_0, \quad 0\leq t\leq t_0 \\ 0 \quad \text{otherwise} \end{cases}
\end{align}
with $x_0 = 0.5$ and $t_0 = 10$. We set the parameters as $\sigma = 0.5$ and $\kappa = 0.5$.

See figure \ref{fig:sim_2} for a comparison to the benchmark. One can see a distinctive dent in all 4 approximations at $x = 0.5$ where the source is discontinuous. The difference between the methods gets smaller as they tend towards zero as the time increases. The $\mathcal{M}_2$ model considerably improves in accuracy compared to $\mathcal{M}_1$ and the spherical harmonics approximations.

\begin{figure}[htp!]%
\includegraphics[width=1\textwidth]{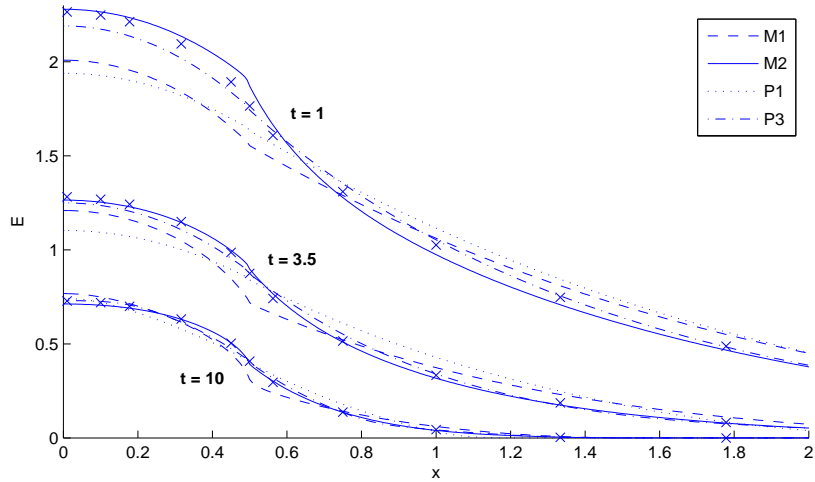}%
\caption{Energy distribution in an infinite domain $x \in \mathbb{R}$; source in $0 \leq x \leq 0.5$ switched on for finite time $0 \leq t \leq 10$. $\sigma = 0.5$, $\kappa = 0.5$, initial temperature $T_0 \approx 0K$.}
\label{fig:sim_2}%
\end{figure}

\section{Conclusion}
Moment methods are an adequate means of treating radiative transfer problems. Utilizing those methods it is possible to approximate the transfer (integro-differential) equation by a relatively simple system of partial differential equations. Still, especially highly anisotropic radiation necessitates a careful choice of the closure, as many models have problems when the diffusive regime is left and the transition regime is entered.

Minimum entropy models are able to handle radiative non-equilibria and provide good approximations that are computationally cheap. The $\mathcal{M}_2$ model corrects the $\mathcal{M}_1$ model's biggest incapability: its inaptitude to distinguish between radiative equilibrium and two identical beams travelling in opposite directions, i.e.\ it cannot handle a zero net flux. By the employment of the second order method no physical shocks occur and the solution always remains positive as opposed to spherical harmonics approximations. Also, the propagation of information is always slower than the speed of light, a fact that does not hold, for instance, in diffusion models. 

Reasonable next steps would be the extension into more spatial dimensions. Then, the moments would cease to be scalar quantities and become tensors. In that framework the analysis of the moment methods and the minimum entropy closure would be much more involved. Nevertheless, only a realistic geometry would allow for modeling inhomogeneities that are non-planar, which are needed to provide accurate methods for radiation therapy especially around tissue including air cavities such as the lungs and the sinus passage. Furthermore it would be very interesting to develop new, or modified, algorithms in order to calculate higher order minimum entropy closures. This would be supported by a more detailed study of the boundary surfaces of higher order Eddington factors.

Additionally, it would be interesting to study higher order minimum entropy methods in different contexts, e.\ g.\ plasma physics.

\section*{Acknowledgments}
This work was supported by German Academic Exchange Service DAAD under grant D/0707534, and by the German Research Foundation DFG under grant KL 1105/14/2.

\bibliographystyle{amsplain}
\bibliography{literatur}

\appendix
\section{Formulae for the $\mathcal{M}_2$ Eddington Factor}
\subsection{The $\mathcal{M}_2$ Eddington factor}
\label{sec:appendix_shi_boson}
\begin{align}
\label{eq:appendix_shi_boson}
\chi (a, b) = &-\frac{1}{8}\frac{\exp(-a)}{b^3 \sqrt{\pi} ( \text{erf} \left ( \frac{1}{2} \frac{a-2b}{\sqrt{-b}} \right ) - \text{erf} \left ( \frac{1}{2} \frac{a+2b}{\sqrt{-b}} \right ) )} \nonumber \\
& \cdot \left ( \exp \left ( \frac{1}{4} \frac{a^2+b^2}{b} \right ) ((-b)^{\frac{3}{2}}( 8 - 4a - 8b) + 2a^2\sqrt{-b}) \right.\nonumber \\
& + a\exp(a) \sqrt{\pi} (a^2-6b) \left ( \text{erf} \left ( \frac{1}{2} \frac{a-2b}{\sqrt{-b}}\right ) + \text{erf} \left ( \frac{1}{2} \frac{a+2b}{\sqrt{-b}}\right ) \right )\nonumber \\
& - 2 \exp \left ( \frac{1}{4} \frac{8ab+a^2+4b^2}{b}  \right ) \left ( (-b)^\frac{3}{2} (4 + 2a - 4b) - 2a^2\sqrt{-b} \right ) 
\end{align}
where we defined the error function erf
\begin{equation}
\text{erf}(x) := \frac{2}{\sqrt{\pi}} \int \limits_0^x \exp(-t^2) dt.
\end{equation}
$a$ and $b$ are the Lagrange-multipliers and have to be determined from the set of constraints $\langle \psi \cdot m \rangle = E$, where $E$ is a vector of prescribed moments.

\subsection{Rational fit to the $\mathcal{M}_2$ Eddington factor}
\label{sec:appendix_shi_fit}
A rational fit to equation (\ref{eq:appendix_shi_boson}) is given by
\begin{align}
\chi_{\text{fit}}(a, b) &= a\cdot   \frac{ k_1 +  k_2\cdot  b - k3\cdot b^2 + k_4\cdot b^3 + k_5\cdot b^4}{k_6 -  k_7\cdot b} \nonumber \\
&+ a^3\cdot (1-b) \frac{(k_8 +  k_9\cdot a^2)\cdot (k_{10} - k_{11}\cdot b + k_{12}\cdot b^2 - k_{13}\cdot b^3)}{k_{14} + k_{15}\cdot b}
\end{align}
with
\begin{table}[h!t]
\centering
\begin{tabular}{|c|c|}
\hline
Parameter & Value \\
\hline
$k_1$ & $0.104107226623813765\cdot 10^{-5}$ \\
$k_2$ & $0.878240820142032813\cdot 10^{-4}$ \\
$k_3$ & $0.834291539331555326\cdot 10^{-4}$\\
$k_4$ & $0.182387623676748914\cdot 10^{-5}$\\
$k_5$ & $0.190298302202491296\cdot 10^{-4}$\\
$k_6$ & $0.372238009572563292\cdot 10^{-4}$\\
$k_7$ & $0.109376202668751704\cdot 10^{-4}$\\
$k_8$ & $-0.397189030190850628\cdot 10^{-3}$\\
$k_9$ & $0.289736270814778120\cdot 10^{-3}$\\
$k_{10}$ & $-3207.68652563260184$\\
$k_{11}$ & $2.77365687089333$\\
$k_{12}$ & $2.63485641674103$\\
$k_{13}$ & $0.576015910414080\cdot 10^{-1}$\\
$k_{14}$ & $-1.17560068853509$\\
$k_{15}$ & $0.485606117598845422$\\
\hline
\end{tabular}
\caption{Parameters of the rational fit to the $\mathcal{M}_2$ Eddington factor}
\label{tab:rational_fit_params}
\end{table}

\end{document}